\documentclass{IEEEtran}

\usepackage{amsmath}
\usepackage{amsthm}
\usepackage{amssymb}
\usepackage{bm}
\usepackage{xspace}
\usepackage{xcolor}
\usepackage{graphicx}
\usepackage{url}
\usepackage{framed}
\usepackage{float}
\usepackage{rotating}
\usepackage{verbatim}
\usepackage{listings}
\usepackage{lscape}
\usepackage[noadjust]{cite}
\usepackage{ifthen}

 \usepackage{pgfplots}
 \usepackage{pgf}
 \usepackage{tikz}
 \usetikzlibrary{arrows,shapes.misc,chains,scopes}
 \pgfplotsset{compat = newest}
 \usepackage{pgfplotstable}
 \usepackage{booktabs}
 \usepackage{colortbl}

\newcommand{\executeiffilenewer}[3]{%
\ifnum\pdfstrcmp{\pdffilemoddate{#1}}%
{\pdffilemoddate{#2}}>0%
{\immediate\write18{#3}}\fi%
}

\graphicspath{{figures/}}

\theoremstyle{plain}
\newtheorem{proposition}{Proposition}

\newtheorem{lemma}{Lemma}

\newtheorem{remark}{Remark}
\theoremstyle{definition}

\newcounter{algocount}
\newcounter{examplecount}

\newenvironment{example}{\refstepcounter{examplecount}\begin{trivlist}\item \textbf{Example \theexamplecount.}}{\end{trivlist}}
\newenvironment{algorithm}[1][]{\refstepcounter{algocount}\setlength{\parindent}{0.5cm}\begin{trivlist}\item \textbf{Algorithm \thealgocount.}#1\\[-0.2cm]\rule{\columnwidth}{1pt}}{\\[-0.2cm]\rule{\columnwidth}{1pt}\end{trivlist}}

\newcommand{\veczero}{\boldsymbol{0}}

\newcommand{\vecc}{\boldsymbol{c}}

\newcommand{\vecp}{\boldsymbol{p}}
\newcommand{\vecq}{\boldsymbol{q}}

\newcommand{\vect}{\boldsymbol{t}}

\newcommand{\vecmu}{\boldsymbol{\mu}}

\newcommand{\bpm}{\begin{pmatrix}}
\newcommand{\epm}{\end{pmatrix}}
\newcommand{\bbm}{\begin{bmatrix}}
\newcommand{\ebm}{\end{bmatrix}}

\DeclareMathOperator*{\argmin}{argmin}

\DeclareMathOperator{\kl}{\mathbb{D}}

\DeclareMathOperator*{\minimize}{minimize}

\DeclareMathOperator*{\st}{subject\;to}

\newcommand{\qtvd}{t^{\mathrm{vd}}}
\newcommand{\vecqtvd}{\boldsymbol{t}^{\mathrm{vd}}}

\newcommand{\qtaid}{t^{\mathrm{a}}}
\newcommand{\vecqtaid}{\boldsymbol{t}^{\mathrm{a}}}

\newcommand{\vecpvd}{\vecp^\mathrm{vd}}

\newcommand{\pvd}{p^\mathrm{vd}}

\newcommand{\KSet}{\mathcal{K}}

\newcommand{\Pmat}{\mathbf{P}}

\newcommand{\klr}{\overline{\kl}}

\usepackage{cite}
\usepackage{url}

\newcommand{\Kvec}{\mathbf{c}}
\newcommand{\lone}[1]{\Vert #1 \Vert_1}
\newcommand{\crev}[1]{\textcolor{black}{#1}}

\title{Greedy Algorithms for Optimal Distribution Approximation}

\author{
\IEEEauthorblockN{
  Bernhard C. Geiger, Georg B\"ocherer}\\
\IEEEauthorblockA{
Institute for Communications Engineering, TU M\"unchen, Germany\\
geiger@ieee.org, georg.boecherer@tum.de}
}

\begin{document}

\maketitle

\begin{abstract}
The approximation of a discrete probability distribution $\vect$ by an $M$-type distribution $\vecp$ is considered. The approximation error is measured by the informational divergence $\kl(\vect\Vert\vecp)$, which is an appropriate measure, e.g., in the context of data compression. Properties of the optimal approximation are derived and bounds on the approximation error are presented, which are asymptotically tight. It is shown that $M$-type approximations that minimize either $\kl(\vect\Vert\vecp)$, or $\kl(\vecp\Vert\vect)$, or the variational distance $\lone{\vecp-\vect}$ can all be found by using specific instances of the same general greedy algorithm.
\end{abstract}


\section{Introduction}
In this work, we consider finite precision representations of probabilistic models. More precisely, if the original model, or \emph{target distribution}, has $n$ non-zero mass points and is given by $\vect:=(t_1,\dots,t_n)$, we wish to approximate it by a distribution $\vecp:=(p_1,\dots,p_n)$, where for every $i$, $p_i=c_i/M$ for some non-negative integer $c_i \le M$. The distribution $\vecp$ is called an \emph{$M$-type distribution}, and the positive integer $M\ge n$ is the \emph{precision} of the approximation. The problem is non-trivial, since computing the numerator $c_i$ by rounding $Mt_i$ to the nearest integer in general fails to yield a distribution.

$M$-type approximations have many practical applications, e.g., in political apportionments, $M$ seats in a parliament need to be distributed to $n$ parties according to the result of some vote $\vect$. This problem led, e.g., to the development of \emph{multiplier methods}~\cite{Dorfleitner_MultiplierMethods}. In communications engineering, example applications are finite precision implementations of probabilistic data compression \cite{rissanen1979arithmetic}, distribution matching \cite{schulte2016constant}, and finite-precision implementations of Bayesian networks~\cite{Druzdzel_BayesianNetworks,Tschiatschek_PrecisionBounds}. In all of these applications, the $M$-type approximation $\vecp$ should be close to the target distribution $\vect$ in the sense of an appropriate error measure. Common choices for this approximation error are the variational distance and the informational divergences:
\begin{subequations}\label{eq:costs}
 \begin{align}
   \lone{\vecp-\vect} &:= \sum_{i=1}^n |p_i-t_i|\label{eq:VD}\\
   \kl(\vecp\Vert\vect)&:=\sum_{i\colon p_i>0} p_i\log\frac{p_i}{t_i}\label{eq:IDsynthesis}\\
   \kl(\vect\Vert\vecp) &:= \sum_{i\colon t_i>0} t_i \log\frac{t_i}{p_i}\label{eq:IDanalysis}
 \end{align}
\end{subequations}
where $\log$ denotes the natural logarithm.

Variational distance and informational divergence~\eqref{eq:IDsynthesis} have been considered by Reznik~\cite{Reznik_Quantization} and B\"{o}cherer~\cite{Boecherer_SCC13}, respectively, who presented algorithms for optimal $M$-type approximation and developed bounds on the approximation error. In a recent manuscript~\cite{Geiger_OptimalQuantization}, we extended the existing works on \eqref{eq:VD} and \eqref{eq:IDsynthesis} to target distributions with infinite support ($n=\infty$) and refined the bounds from~\cite{Reznik_Quantization,Boecherer_SCC13}.

In this work, we focus on the approximation error \eqref{eq:IDanalysis}. It is an appropriate cost function for data compression \cite[Thm.~5.4.3]{Cover_Information2} and seems apropriate for the approximation of parameters in Bayesian networks (see Sec.~\ref{sec:outlook}). Nevertheless, to the best of the authors' knowledge, the characterization of $M$-type approximations minimizing $\kl(\vect\Vert\vecp)$ has not received much attention in literature so far.

Our contributions are as follows. In Sec.~\ref{sec:greedy}, we present an efficient greedy algorithm to find $M$-type distributions minimizing \eqref{eq:IDanalysis}. We then discuss in Sec.~\ref{sec:analysis} the properties of the optimal $M$-type approximation and bound the approximation error \eqref{eq:IDanalysis}. Our bound incorporates a reverse Pinsker inequality recently suggested in~\cite[Thm.~7]{Verdu_ITABound}. The algorithm we present is an instance of a greedy algorithm similar to \emph{steepest ascent hill climbing}~\cite[Ch.~2.6]{Michalewicz_ModernHeuristics}. As a byproduct, we unify this work with~\cite{Reznik_Quantization,Boecherer_SCC13,Geiger_OptimalQuantization} by showing that also the algorithms optimal w.r.t. variational distance \eqref{eq:VD} and informational divergence~\eqref{eq:IDsynthesis} are instances of the same general greedy algorithm, see Sec.~\ref{sec:greedy}.

\section{Greedy Optimization}\label{sec:greedy}
In this section, we define a class of problems that can be optimally solved by a greedy algorithm. Consider the following example:

\begin{example}\label{ex:jobs}
  Suppose there are $n$ queues with jobs, and you have to select $M$ jobs minimizing the total time spent. A greedy algorithm suggests to select always the job with the shortest duration, among the jobs that are at the front of their queues. If the jobs in each queue are ordered by increasing duration, then this greedy algorithm is optimal.
\end{example}

We now make this precise: Let $M$ be a positive integer, e.g.,  the number of jobs that have to be completed, and let $\delta_i{:}\ \mathbb{N}\to\mathbb{R}$, $i=1,\dotsc,n$, be a set of functions, e.g., $\delta_i(k)$ is the duration of the $k$-th job in the $i$-th queue. Let furthermore $\Kvec_0:=(c_{1,0},\dots,c_{n,0})\in\mathbb{N}_0^n$ be a \emph{pre-allocation}, representing a constraint that has to be fulfilled (e.g., in each queue at least one job has to be completed) or a chosen initialization. Then, the goal is to minimize 
\begin{equation}
 U(\Kvec) := \sum_{i=1}^n\sum_{k_i=c_{i,0}+1}^{c_i} \delta_i(k_i) \label{eq:general:cost}
\end{equation}
i.e., to find a \emph{final allocation} $\Kvec:=(c_1,\cdots,c_n)$ satisfying, for all $i$, $c_i\ge c_{i,0}$ and $\lone{\Kvec}=M$. A greedy method to obtain such a final allocation is presented in Algorithm~\ref{alg:GeneralAlgo}. We show in Appendix~\ref{proof:optAlgo} that this algorithm is optimal if the functions $\delta_i$ satisfy certain conditions:

\begin{figure}
\begin{algorithm}[~Greedy Algorithm]\label{alg:GeneralAlgo} 
 \\
Initialize $k_i=c_{i,0}$, $i=1,\dotsc,n$.\\
\textbf{repeat} $M-\lone{\Kvec_0}$ times\\
\indent Compute $\delta_i(k_i+1)$, $i=1,\dotsc,n$.\\
\indent Compute $j=\min\argmin_i\delta_i(k_i+1)$. \\ \indent\indent\crev{// (choose one minimal element)}\\
\indent Update $k_j\leftarrow k_j+1$.\\
\textbf{end repeat}\\
Return $\Kvec=(k_1,\dots,k_n)$.
\end{algorithm}
\end{figure}

\begin{proposition}\label{prop:optAlgo}
 If the functions $\delta_i(k)$ are non-decreasing in $k$, Algorithm~\ref{alg:GeneralAlgo} achieves a global minimum $U(\Kvec)$ for a given pre-allocation $\Kvec_0$ and a given $M$.
\end{proposition}

\begin{remark}
 \crev{The minimum of $U(\Kvec)$ may not be unique.}
\end{remark}

\begin{remark}\label{rem:convex}
 If a function $f_i{:}\ \mathbb{R}\to\mathbb{R}$ is convex, the difference $\delta_i(k)=f_i(k)-f_i(k-1)$ is non-decreasing in $k$. Hence, Algorithm~\ref{alg:GeneralAlgo} also minimizes
  \begin{equation}\label{eq:convexsum}
  U(\Kvec) = \sum_{i=1}^n f_i(c_i).
 \end{equation}
\end{remark}

Remark~\ref{rem:convex} connects Algorithm~\ref{alg:GeneralAlgo} to steepest ascent hill climbing~\cite[Ch.~2.6]{Michalewicz_ModernHeuristics} with fixed step size and a constrained number of $M$ steps. Hill climbing is optimal for convex problems, suggesting an interesting connection with Proposition~\ref{prop:optAlgo}.

We now show that instances of Algorithm~\ref{alg:GeneralAlgo} can find $M$-type approximations $\vecp$ minimizing each of the cost functions in~\eqref{eq:costs}. Noting that $p_i=c_i/M$ for some non-negative integer $c_i$, we can rewrite the cost functions as follows:
\begin{subequations}\label{eq:costs_rewritten}
 \begin{align}
   \lone{\vecp-\vect} &= \frac{1}{M}\sum_{i=1}^n |c_i-Mt_i|\label{eq:VD_rewritten}\\
   \kl(\vecp\Vert\vect)&=\frac{1}{M}\left(\sum_{i\colon c_i>0} c_i\log\frac{c_i}{t_i}\right) - \log M\label{eq:IDsynthesis_rewritten}\\
   \kl(\vect\Vert\vecp) &= \log M - H(\vect) - \sum_{i\colon t_i>0} t_i \log c_i.\label{eq:IDanalysis_rewritten}
 \end{align}
\end{subequations}
Ignoring constant terms, these cost functions are all instances of Remark~\ref{rem:convex} for convex functions $f_i{:}\ \mathbb{R}\to\mathbb{R}$ (see Table~\ref{tab:summary}). Hence, the three different $M$-type approximation problems set up by~\eqref{eq:costs} can all be solved by instances of Algorithm~\ref{alg:GeneralAlgo}, for a trivial pre-allocation $\Kvec_0=\veczero$ and after taking $M$ steps: The final allocation $\Kvec$ simply defines the $M$-type approximation by $p_i=c_i/M$. For variational distance optimal approximation, we showed in~\cite[Lem.~3]{Geiger_OptimalQuantization} that every optimal $M$-type approximation satisfies $p_i\ge \lfloor Mt_i\rfloor/M$, hence one may speed up the algorithm by pre-allocating $c_{i,0}=\lfloor Mt_i\rfloor$. We furthermore show in Lemma~\ref{lem:A:Element} below that the support of the optimal $M$-type approximation in terms of~\eqref{eq:IDanalysis} equals the support of $\vect$ (if $M$ is large enough). Assuming that $\vect$ is positive, one can pre-allocate the algorithm with $c_{i,0}=1$. We summarize these instantiations of Algorithm~\ref{alg:GeneralAlgo} in Table~\ref{tab:summary}.

\begin{table}[t]
\caption{Instances of Algorithm~\ref{alg:GeneralAlgo} Optimizing~\eqref{eq:costs}.}
\label{tab:summary}
\centering
 \begin{tabular}{c||m{1.5cm}|m{2.25cm}|>{\centering\arraybackslash}m{0.75cm}|m{0.8cm}}
  Cost & $f_i(x)$ & $\delta_i(k)$ & $c_{i,0}$ & Refs\\
  \hline
  $\lone{\vecp-\vect}$ %
  & $|x-Mt_i|$%
  & $|k-Mt_i|$ \newline $\quad-|k-1-Mt_i|$%
  & $\lfloor M t_i \rfloor$%
  & \cite{Reznik_Quantization,Geiger_OptimalQuantization}\\
  \hline
  $\kl(\vecp\Vert\vect)$%
  & $x\log(x/t_i)$%
  & {$k\log\frac{k}{k-1}$
    \newline
    $+\log(k-1)\newline-\log t_i$}%
  & $0$%
  & \cite{Boecherer_SCC13,Geiger_OptimalQuantization}\\
  \hline
  $\kl(\vect\Vert\vecp)$%
  & $-t_i\log x$%
  & $t_i\log((k-1)/k)$%
  & $\lceil t_i \rceil$%
  & This \newline work\\
  \hline
 \end{tabular}
\end{table}

This list of instances of Algorithm~\ref{alg:GeneralAlgo} minimizing information-theoretic or probabilistic cost functions can be extended. For example, the $\chi^2$-divergences $\chi^2(\vect||\vecp)$ and $\chi^2(\vecp||\vect)$ can also be minimized, since the functions inside the respective sums are convex. However, R\'{e}nyi divergences of orders $\alpha\neq 1$ cannot be minimized by applying Algorithm~\ref{alg:GeneralAlgo}.

\section{$M$-Type Approximation Minimizing $\kl(\vect\Vert\vecp)$: Properties and Bound}\label{sec:analysis}
As shown in the previous section, Algorithm~\ref{alg:GeneralAlgo} presents a minimizer of the problem $\min_{\vecp} \kl(\vect\Vert\vecp)$ if instantiated according to Table~\ref{tab:summary}. Let us call this minimizer $\vecqtaid$. Recall that $\vect$ is positive and that $M\ge n$. The support of $\vecqtaid$ must be at least as large as the support of $\vect$, since otherwise $\kl(\vect\Vert\vecqtaid)=\infty$. Note further that the costs $\delta_i(k)$ are negative if $t_i>0$ and zero if $t_i=0$; hence, if $t_i=0$, the index $i$ cannot be chosen by Algorithm~\ref{alg:GeneralAlgo}, thus also $\qtaid_i=0$. This proves

\begin{lemma}\label{lem:A:Element}
If $M\ge n$, the supports of $\vect$ and $\vecqtaid$ coincide, i.e., $t_i=0\Leftrightarrow \qtaid_i=0$.
\end{lemma}

The assumption that $\vect$ is positive and that $M\ge n$ hence comes without loss of generality. In contrast, neither variational distance nor informational divergence~\eqref{eq:IDsynthesis} require $M\ge n$: As we show in~\cite{Geiger_OptimalQuantization}, the $M$-type approximation problem remains interesting even if $M<n$.

Lemma~\ref{lem:A:Element} explains why the optimal $M$-type approximation does not necessarily result in a ``small'' approximation error:
\begin{example}
Let $\vect=(1-\varepsilon,\frac{\varepsilon}{n-1},\dots,\frac{\varepsilon}{n-1})$ and $M=n$, hence by Lemma~\ref{lem:A:Element}, $\vecqtaid=\frac{1}{n}(1,1,\dots,1)$. It follows that $\kl(\vect\Vert\vecqtaid)=\log n - H(\vect)$, which can be made arbitrarily close to $\log n$ by choosing a small positive $\varepsilon$.
\end{example}

In Table~\ref{tab:summary} we made use of \cite[Lem.~3]{Geiger_OptimalQuantization}, which says that every $\vecp$ minimizing the variational distance $\lone{\vecp-\vect}$ satisfies $p_i\ge \lfloor Mt_i\rfloor/M$, to speed up the corresponding instance of Algorithm~\ref{alg:GeneralAlgo} by proper pre-allocation. Initialization by rounding is not possible when minimizing $\kl(\vect\Vert\vecp)$, as shown in the following two examples:

\begin{example}\label{ex:A:noMEntry}\label{ex:A:noRoundOff}
 Let $\vect=(17/20, 3/40, 3/40)$ and $M=20$. The optimal $M$-type approximation is $\vecp=(8/10,1/10,1/10)$, hence $p_1<\lfloor Mt_1\rfloor/M$. Initialization via rounding off fails.
\end{example}

\begin{example}\label{ex:A:noRoundUp}
 Let $\vect=(0.719, 0.145, 0.088, 0.048)$ and $M=50$. The optimal $M$-type approximation is $\vecp=(0.74,    0.14,    0.08,   0.04)$, hence $p_1>\lceil Mt_1\rceil/M$. Initialization via rounding up fails.
\end{example}

To show that informational divergence vanishes for $M\to\infty$, assume that $M>1/t_i$ for all $i$. Hence, the variational distance optimal approximation $\vecqtvd$ has the same support as $\vect$, which ensures that $\kl(\vect\Vert \vecqtvd)<\infty$. By similar arguments as in the proof of~\cite[Prop.~4.1)]{Geiger_OptimalQuantization}, we obtain
 \begin{equation}
  \kl(\vect\Vert\vecqtaid) \le \kl(\vect\Vert \vecqtvd) \le\log\left(1+\frac{n}{2M}\right) \stackrel{M\to\infty}{\longrightarrow}0.\label{eq:boundComparison}
 \end{equation}

We now develop an upper bound on $\kl(\vect\Vert\vecqtaid)$ that holds for every $M$. To this end, we first approximate $\vect$ by a distribution $\vect^*$ in $\mathcal{P}_M:=\{\vecp\colon\forall i{:}\ p_i\ge 1/M, \Vert\vecp\Vert_1=1\}$ that minimizes $\kl(\vect\Vert\vect^*)$. If $\vect^*$ is unique, then it is called the \emph{reverse $I$-projection}~\cite[Sec.~I.A]{Csiszar_IProjections} of $\vect$ onto $\mathcal{P}_M$. Since $\vect^*\in\mathcal{P}_M$, its variational distance optimal approximation $\vecqtvd$ has the same support as $\vect$, which allows us to bound $\kl(\vect\Vert\vecqtaid)$ by $\kl(\vect\Vert\vecqtvd)$.

\begin{lemma}\label{lem:Iproj}
Let $\vect^*\in\mathcal{P}_M$ minimize $\kl(\vect\Vert\vect^*)$. Then,
 \begin{equation}
  t^*_i := \frac{t_i}{\nu(M)} + \left(\frac{1}{M}-\frac{t_i}{\nu(M)}\right)^+\label{eq:vdef}
 \end{equation}
where $\nu(M)$ is such that $\Vert\vect^*\Vert_1=1$, and where $(x)^+ := \max\{0,x\}$.
\end{lemma}

\begin{proof}
 See Appendix~\ref{proof:Iproj}.
\end{proof}

 Let $\KSet:=\{i: t_i<\nu(M)/M\}$, $k:=|\KSet|$, and $T_{\KSet}:=\sum_{i\in\KSet} t_i$. The parameter $\nu$ must scale the mass $(1-T_{\KSet})$ such that it equals $(M-k)/M$, i.e., we have
\begin{align}\label{eq:boundOnNu}
 \nu = \frac{1-T_{\KSet}}{1-\frac{k}{M}}. 
\end{align}
If, for all $i$, $t_i>1/M$, then $\vect\in\mathcal{P}_M$, hence $\vect^*=\vect$ is feasible and $\nu(M)=1$. One can show that $\nu(M)$ decreases with $M$.

\begin{proposition}[Approximation Bounds]\label{prop:A:Bound}
 \begin{align}
  \kl(\vect\Vert\vecqtaid) &\le \log\nu(M)+\frac{\log(2)}{2}\left(1-\nu(M)\left(1-\frac{n}{M}\right)\right) \label{eq:A:Bound}
 \end{align}
\end{proposition}

\begin{proof}
 See Appendix~\ref{proof:A:bound}.
\end{proof}

The first term on the right-hand side of \eqref{eq:A:Bound} accounts for the error caused by first approximating $\vect$ by $\vect^*$ (in the sense of Lemma~\ref{lem:Iproj}). The second term then accounts for the additional error caused by the $M$-type approximation of $\vect^*$. The bound incorporates the reverse Pinsker inequality~\cite[Thm.~7]{Verdu_ITABound}, see Appendix~\ref{proof:A:bound}. If $M>t_i$ for every $i$, hence $\vect\in\mathcal{P}_M$, then $\nu(M)=1$ and only the second term remains. For large $M$,~\eqref{eq:A:Bound} yields better results than~\eqref{eq:boundComparison}, justifying this bound on the whole range of $M$. We illustrate the bounds for an example in Fig.~\ref{fig:bounds}.

\section{Outlook}\label{sec:outlook}
A possible application for our algorithms is the $M$-type approximation of Markov models, i.e., approximating the transition matrix $\mathbf{T}$ of an $n$-state, irreducible Markov chain with invariant distribution vectors $\vecmu$ by a transition matrix $\Pmat$ containing only $M$-type probabilities. Generalizing~\eqref{eq:IDanalysis}, the approximation error can be measured by the informational divergence rate~\cite{Rached_KLDR}
\begin{equation}
 \klr(\mathbf{T}\Vert\Pmat) := \sum_{i,j=1}^n \mu_iT_{ij} \log\frac{T_{ij}}{P_{ij}} 
 = \sum_{i=1}^n \mu_i \kl(\vect_i\Vert\vecp_i)
\end{equation}
The optimal $M$-type approximation is found by applying the instance of Algorithm~\ref{alg:GeneralAlgo} to each row separately, and Lemma~\ref{lem:A:Element} ensures that the transition graph of $\Pmat$ equals that of $\mathbf{T}$, i.e., the approximating Markov chain is irreducible. Future work shall extend this analysis to hidden Markov models and should investigate the performance of these algorithms in practical scenarios, e.g., in speech processing.

\pgfplotsset{select coords between index/.style 2 args={
    x filter/.code={
        \ifnum\coordindex<#1\def\pgfmathresult{}\fi
        \ifnum\coordindex>#2\def\pgfmathresult{}\fi
    }
}}
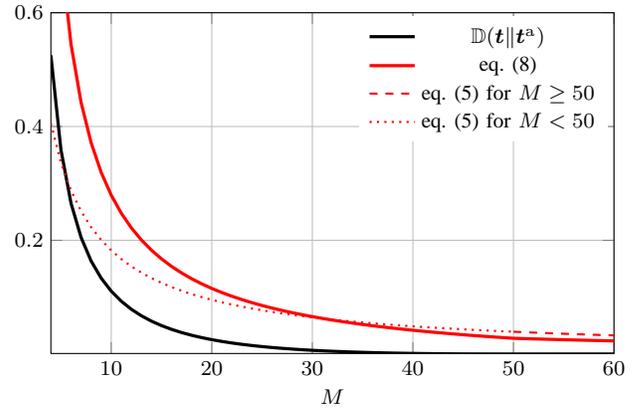
\begin{figure}[t]
 \centering
 \footnotesize
 \begin{tikzpicture}
\begin{axis}[
width=0.5\textwidth,
height=0.25\textheight,
xmin = 4,
ymin = 0.001,
xmax = 60,
ymax = 0.6,
xlabel={$M$},
ylabel={},
grid=both,
xlabel near ticks,
ylabel near ticks,
legend columns=1,
legend style={at={(0.99,0.99)},anchor=north east,draw=none},
legend entries = {{$\kl(\vect\Vert\vecqtaid)$},{eq.~\eqref{eq:A:Bound}},{eq.~\eqref{eq:boundComparison} for $M\ge 50$},{eq.~\eqref{eq:boundComparison} for $M< 50$}},
]
\addplot[black,no markers,very thick,solid]
table{true.dat};
\addplot[red,no markers,very thick,solid]
table{bound.dat};
\addplot[red,no markers, thick,solid,dashed,select coords between index={46}{56}]
table{no_bound.dat};
\addplot[red,no markers, thick,solid,dotted,select coords between index={0}{46}]
table{no_bound.dat};
\end{axis}
\end{tikzpicture}
 \caption{Evaluating the bounds~\eqref{eq:A:Bound} and~\eqref{eq:boundComparison} for $\vect=(0.48,0.48,0.02,0.02)$. Note that~\eqref{eq:boundComparison} is a valid bound only for $M\ge 50$, i.e., where the curve is dashed.}
 \label{fig:bounds}
\end{figure}

Another possible application is the approximation of Bayesian network parameters. The authors of~\cite{Druzdzel_BayesianNetworks} approximated the true parameters using a stationary multiplier method from~\cite{Heinrich_Multiplier}. Since rounding probabilities to zero led to bad classification performance, they replaced zeros in the approximating distribution afterwards by small values; thus, they also approximated true zero probabilities by positive ones. We believe that these problems can be removed by instantiating Algorithm~\ref{alg:GeneralAlgo} for cost~\eqref{eq:IDanalysis}. This automatically prevents approximating non-zero probabilities with zeros and vice-versa, see Lemma~\ref{lem:A:Element}.

Recent work suggested rounding \emph{log-probabilities}, i.e., to approximate $\log t_i$ by $\log p_i = -c_i/M$ for a non-negative integer $c_i$~\cite{Tschiatschek_PrecisionBounds}. Finding an optimal approximation that corresponds to a true distribution is equivalent to solving
\begin{align*}
 &\min\ d(\vect,\vecp)\\
 & \text{s.t.} \ \Vert e^{-\vecc} \Vert_{1/M}=1
\end{align*}
where $d(\cdot,\cdot)$ denotes any of the considered cost functions~\eqref{eq:costs}. If $M=1$ and $d(\vect,\vecp)=\kl(\vect\Vert\vecp)$ using the binary logarithm, the constraint translates to the requirement that $\vect$ is approximated by a complete binary tree. Then, the optimal approximation is the Huffman code for $\vect$.

\section*{Acknowledgements}
The work of Bernhard C. Geiger was partially funded by the Erwin Schr\"odinger Fellowship J 3765 of the Austrian Science Fund. The work of Georg B\"ocherer was partly supported by the German Ministry of Education and Research in the framework of an Alexander von Humboldt Professorship.

\appendix
\section{Proofs}

\subsection{Proof of Proposition~\ref{prop:optAlgo}}\label{proof:optAlgo}
 Since a pre-allocation only fixes a lower bound for $U(\Kvec)$, w.l.o.g.\ we assume that $\Kvec_{0}=\veczero$ and thus $\Kvec\in\mathbb{N}_0^n$ with $\lone{\Kvec}=M$. Consider the set $\mathcal{D}:=\{\delta_i(k_i){:}\ k_i\in\mathbb{N},i=1,\dots,n\}$ and assume that the (not necessarily unique) set $\mathcal{D}_M$ consists of $M$ smallest values in $\mathcal{D}$, i.e., $|\mathcal{D}_M|=M$ and
 \begin{equation}\label{eq:smallset}
  \forall d\in\mathcal{D}_M,d'\in\mathcal{D}\setminus\mathcal{D}_M{:} \quad d\le d'.
 \end{equation}
 
 Clearly, $U(\Kvec)$ cannot be smaller than the sum over all elements in $\mathcal{D}_M$.  Since the $\delta_i$ are non-decreasing, there exists at least one final allocation $\Kvec$ that takes successively the first $c_i$ values from each queue $i$, i.e., $\mathcal{D}_M=\{\delta_1(1),\dots,\delta_1(c_1),\dots,\delta_n(1),\dots,\delta_n(c_n)\}$ satisfies~\eqref{eq:smallset}. This shows that the lower bound induced by \eqref{eq:smallset} can actually be achieved.
 
 \crev{We prove the optimality of Algorithm~\ref{alg:GeneralAlgo} by contradiction:}
 Assume that Algorithm~\ref{alg:GeneralAlgo} finishes with a final allocation $\tilde{\Kvec}$ such that $U(\tilde{\Kvec})$ is strictly larger than the \crev{(unique)} sum over all elements in \crev{(non-unique)} $\mathcal{D}_M$. Hence, $\tilde\Kvec$ must exchange at least one of the elements in $\mathcal{D}_M$ for an element that is strictly larger. Thus, by the properties of the functions $\delta_i$ and Algorithm~\ref{alg:GeneralAlgo}, there must be indices $\ell$ and $m$ such that $\tilde{c}_\ell>c_\ell$, $\tilde{c}_m<c_m$, and $\delta_\ell(\tilde{c}_\ell)\ge \delta_\ell(c_\ell+1)>\delta_m(c_m)\ge\delta_m(\tilde{c}_m)$. At each iteration of the algorithm, the current allocation at index $m$ satisfies $k_m\le \tilde{c}_m<c_m$. Since $\delta_m(c_m)<\delta_\ell(c_\ell+1)$, $\delta_\ell(c_\ell+1)$ can never be a minimal element, and hence is not chosen by Algorithm~\ref{alg:GeneralAlgo}. This contradicts the assumption that Algorithm~\ref{alg:GeneralAlgo} finishes with a $\tilde{\Kvec}$ such that $U(\tilde{\Kvec})$ is strictly larger than the sum of $\mathcal{D}$'s $M$ smallest values. \qed

\subsection{Proof of Lemma~\ref{lem:Iproj}} \label{proof:Iproj}
The problem finding \crev{a $\vect^*\in\mathcal{P}_M$ minimizing $\kl(\vect\Vert\vect^*)$ is equivalent to finding an optimal point of the problem:}  
\begin{subequations}
 \begin{align}
\minimize_{\vecp\in\mathbb{R}^n_{>0}}	&\quad -\sum_{i=1}^n t_i\log p_i\\	
\st			&\quad  \frac{1}{M}-p_i\leq 0,\quad i=1,2,\dotsc,n\label{eq:cond1}\\
			&\quad -1+\sum_{i=1}^n p_i=0
\end{align}
\end{subequations}
The Lagrangian of the problem is
\begin{multline}
L(\vecp,\boldsymbol{\lambda},\nu)=-\sum_{i=1}^n t_i\log p_i\\+\sum_{i=1}^n \lambda_i\left(\frac{1}{M}-p_i\right)+\nu\left(-1+\sum_{i=1}^n p_i\right).
\end{multline}
By the KKT conditions \cite[Ch.~5.5.3]{Boyd_CVX04}, a feasible point $\vect^*$ is optimal if, for every $i=1,\dots,n$,
\begin{subequations}
 \begin{align}
\lambda_i&\geq 0\\
\lambda_i\left(\frac{1}{M}-t_i^*\right)&=0\label{eq:slackness}\\
\frac{\partial}{\partial p_i}L(\vecp,\boldsymbol{\lambda},\nu)\vert_{\vecp=\vect^*}=-\frac{t_i}{t_i^*}-\lambda_i+\nu&=0\label{eq:optimal}
\end{align}
\end{subequations}
By \eqref{eq:cond1}, we have $t_i^*\geq 1/M$. If $t_i^*>1/M$, then $\lambda_i=0$ by~\eqref{eq:slackness} and $t_i^*=t_i/\nu$ by~\eqref{eq:optimal}. Thus
\begin{align}
t_i^* = \frac{t_i}{\nu}+\left(\frac{1}{M}-\frac{t_i}{\nu}\right)^+
\end{align}
where $\nu$ is such that $\sum_{i=1}^n t_i^*=1$.\qed

\subsection{Proof of Proposition~\ref{prop:A:Bound}} \label{proof:A:bound}
Reverse $I$-projections admit a Pythagorean inequality~\cite[Thm.~1]{Csiszar_IProjections}. In other words, if $\vecp$ is a distribution, $\vecp^*$ its reverse $I$-projection onto a set $\mathcal{S}$, and $\vecq$ any distribution in $\mathcal{S}$, then
\begin{equation}
 \kl(\vecp\Vert\vecq) \ge \kl(\vecp\Vert\vecp^*) + \kl(\vecp^*\Vert\vecq).
\end{equation}
For the present scenario, we can show an even stronger result:

\begin{lemma}\label{lem:pythagoras}
 Let $\vect$ be the target distribution, let $\vect^*$ be as in Lemma~\ref{lem:Iproj}, and let $\vecqtvd$ be the variational distance optimal $M$-type approximation of $\vect^*$. Then,
 \begin{equation}
  \kl(\vect\Vert\vecqtvd) = \kl(\vect\Vert\vect^*) + \nu\kl(\vect^*\Vert\vecqtvd).
 \end{equation}
\end{lemma}

\begin{proof}
 \begin{align}
    \kl(\vect\Vert\vecqtvd) &= \sum_{i=1}^n t_i\log\frac{t_i}{\qtvd_i}\\
    &= \sum_{i=1}^n t_i\log\frac{t_i t^*_i}{\qtvd_i t^*_i}\\
    &= \sum_{i=1}^n t_i\log\frac{t_i}{t^*_i} + \sum_{i=1}^n t_i\log\frac{t^*_i}{\qtvd_i}\\
    &\stackrel{(a)}{=} \kl(\vect\Vert\vect^*) + \nu\sum_{i\notin\KSet} \frac{t_i}{\nu}\log\frac{t^*_i}{\qtvd_i}\\
    &\stackrel{(b)}{=} \kl(\vect\Vert\vect^*)+ \nu\kl(\vect^*\Vert\vecqtvd)
 \end{align}
where $(a)$ follows because for $i\in\KSet$, $t^*_i=\qtvd_i=1/M$ and $(b)$ is because for $i\notin\KSet$, $t^*_i=t_i/\nu$.
\end{proof}

We now bound the summands in Lemma~\ref{lem:pythagoras}.

\begin{lemma}\label{lem:boundIDviaVD}
In the setting of Lemma~\ref{lem:pythagoras},
 \begin{equation}
  \kl(\vect^*\Vert\vecqtvd)\le \log(2)\lone{\vect^*-\vecqtvd}.
 \end{equation}
\end{lemma}

\begin{proof}
 We first emply a \emph{reverse Pinsker inequality} from~\cite[Thm.~7]{Verdu_ITABound}, stating that
 \begin{equation}
  \kl(\vect^*\Vert\vecqtvd)\le \frac{1}{2}\frac{r\log r}{r-1} \lone{\vect^*-\vecqtvd}
 \end{equation}
 where $r:=\sup_{i\colon t_i^*>0} \frac{t_i^*}{\qtvd_i}$. Furthermore, since for variational distance optimal approximations we always have $|t_i^*-\qtvd_i|<1/M$~\cite[Lem.~3]{Geiger_OptimalQuantization}, we can bound
 \begin{equation}
  r < \frac{\qtvd_i+\frac{1}{M}}{\qtvd_i} \le 2
 \end{equation}
 since $\qtvd_i\ge\lfloor Mt_i^*\rfloor/M\ge 1/M$.
\end{proof}

\begin{lemma}\label{lem:nubound}
 In the setting of Lemma~\ref{lem:pythagoras},
 \begin{equation}
  \kl(\vect\Vert\vect^*) \le \log \nu.
 \end{equation}
\end{lemma}

\begin{proof}
 \begin{align}
  \kl(\vect\Vert\vect^*) &=  \sum_{i=1}^n t_i\log\frac{t_i}{t^*_i}\\
  &= \sum_{i\notin\KSet} t_i\log\frac{\nu t_i}{t_i} + \sum_{i\in\KSet} t_i\log Mt_i\\
  &\stackrel{(a)}{\le} (1-T_{\KSet})\log\nu + \sum_{i\in\KSet} t_i\log \nu\\
  & = \log \nu
 \end{align}
where $(a)$ is because for $i\in\KSet$, $Mt_i\le\nu$.
\end{proof}

To bound $\lone{\vect^*-\vecqtvd}$, we present
\begin{lemma}\label{lem:vdbound}
 Let $\vecp^*$ be a sub-probability distribution with $m\le M$ masses and total weight $1-T$, and let ${\vecpvd}^*$ be the variational distance optimal $M$-type approximation using $J\le M$ masses. Then,
 \begin{equation}\label{eq:tvdbound}
  \Vert\vecp^*-{\vecpvd}^* \Vert_1 \le \frac{m}{2M}+\frac{(M-MT-J)^2}{2mM}.
 \end{equation}
\end{lemma}

Note that for $J=M$ we recover~\cite[Lemma~4]{Geiger_OptimalQuantization}.

\begin{proof}
Assume that either $\forall i{:}\ p_i^*\ge{\pvd_i}^*$ or $\forall i{:}\ p_i^*\le{\pvd_i}^*$. Note that this is possible since $\vecp^*$ and ${\vecpvd}^*$ are sub-probability distributions. Then, $\lone{\vecp^*-{\vecpvd}^*}=|1-T-J/M|$ which satisfies this bound. This can be seen by rearranging~\eqref{eq:tvdbound} such that $J$ only appears on the left-hand side; the maximizing $J$ (not necessarily integer) then satisfies~\eqref{eq:tvdbound} with equality. 

We thus remain to treat the case where after rounding off all indices, $1\le L\le M-1$ masses remain and we have
 \begin{equation}
  \sum_{i=1}^m p_i^*-\frac{\lfloor M p_i^* \rfloor}{M} =: \sum_{i=1}^m e_i = 1-T-\frac{J-L}{M} =: g(L).
 \end{equation}
 The variational distance is minimized by distributing the $L$ masses to $L$ indices $i\in\mathcal{L}$ with the largest errors $e_i$, hence
\begin{align}
 \Vert\vecp^*-{\vecpvd}^* \Vert_1 &= \sum_{i\in\mathcal{L}} \left(\frac{1}{M}-e_i\right) +\sum_{i\notin\mathcal{L}} e_i\\
 &\stackrel{(a)}{\le} \frac{L}{M}-\frac{L}{n} g(L) + \frac{n-L}{n} g(L)
\end{align}
where $(a)$ follows because for $i\in\mathcal{L},j\notin\mathcal{L}$, $e_i\ge e_j$. This is maximized for $L=\frac{n-(M-MT-J)}{2}$ (not necessarily integer), which after inserting yields the upper bound.
\end{proof}

\begin{proof}[Proof of Bound in Proposition~\ref{prop:A:Bound}]
We start by bounding the informational divergence $\kl(\vect\Vert\vecqtaid)$ by the informational divergence between $\vect$ and the variational distance optimal approximation $\vecqtvd$ of its reverse $I$-projection $\vect^*$ onto $\mathcal{P}_M$:
 \begin{align}
 \kl(\vect\Vert\vecqtaid) &\le \kl(\vect\Vert\vecqtvd)\\
 &\stackrel{(a)}{=} \kl(\vect\Vert\vect^*)+ \nu\kl(\vect^*\Vert\vecqtvd)\\
 &\stackrel{(b)}{\le} \log\nu + \nu\log(2)\lone{\vect^*-\vecqtvd}\\
 &\stackrel{(c)}{\le} \log\nu + \nu\log(2)\frac{n-k}{2M}\\
 &\stackrel{(d)}{\le} \log\nu + \nu\log(2)\frac{n-M+\frac{M}{\nu}}{2M}\\
 &=\log\nu + \frac{\log(2)}{2}\left(1-\nu\left(1-\frac{n}{M}\right)\right)
\end{align}
where
\begin{enumerate}
 \item[$(a)$] is due to Lemma~\ref{lem:pythagoras},
 \item[$(b)$] is due to Lemmas~\ref{lem:boundIDviaVD} and~\ref{lem:nubound},
 \item[$(c)$] is due to Lemma~\ref{lem:vdbound} with $m=n-k$, $1-T=1-k/M$, and $J=M-k$, and
 \item[$(d)$] is follows by bounding $k$ from below via~\eqref{eq:boundOnNu}
 \begin{equation}
    k=\frac{M}{\nu} (\nu-1+T_{\KSet}) \ge \frac{M}{\nu}(\nu-1) = M-\frac{M}{\nu}.
 \end{equation}
\end{enumerate}
\end{proof}

\bibliographystyle{IEEEtran}
\bibliography{%
../../../References/IEEEabrv,%
../../../References/myOwn,%
../../../References/textbooks,%
../../../References/ProbabilityPapers,%
../../../References/HMMRate,%
../../../References/MarkovStuff,%
../../../References/LumpingTanja,%
../../../References/ITBasics,%
../../../References/ITAlgos}

\begin{thebibliography}{10}
\providecommand{\url}[1]{#1}
\csname url@samestyle\endcsname
\providecommand{\newblock}{\relax}
\providecommand{\bibinfo}[2]{#2}
\providecommand{\BIBentrySTDinterwordspacing}{\spaceskip=0pt\relax}
\providecommand{\BIBentryALTinterwordstretchfactor}{4}
\providecommand{\BIBentryALTinterwordspacing}{\spaceskip=\fontdimen2\font plus
\BIBentryALTinterwordstretchfactor\fontdimen3\font minus
  \fontdimen4\font\relax}
\providecommand{\BIBforeignlanguage}[2]{{%
\expandafter\ifx\csname l@#1\endcsname\relax
\typeout{** WARNING: IEEEtran.bst: No hyphenation pattern has been}%
\typeout{** loaded for the language `#1'. Using the pattern for}%
\typeout{** the default language instead.}%
\else
\language=\csname l@#1\endcsname
\fi
#2}}
\providecommand{\BIBdecl}{\relax}
\BIBdecl

\bibitem{Dorfleitner_MultiplierMethods}
G.~Dorfleitner and T.~Klein, ``Rounding with multiplier methods: {An} efficient
  algorithm and applications in statistics,'' \emph{Statistical Papers},
  vol.~40, pp. 143--157, 1999.

\bibitem{rissanen1979arithmetic}
J.~Rissanen and G.~G. Langdon, ``Arithmetic coding,'' \emph{IBM J. Res.
  Develop.}, vol.~23, no.~2, pp. 149--162, 1979.

\bibitem{schulte2016constant}
P.~Schulte and G.~B\"ocherer, ``Constant composition distribution matching,''
  \emph{{IEEE} Trans. Inf. Theory}, vol.~62, no.~1, pp. 430--434, Jan. 2016.

\bibitem{Druzdzel_BayesianNetworks}
M.~J. Dru\.{z}d\.{z}el and A.~Oni\'{s}ko, ``Are {Bayesian} networks sensitive
  to precision of their parameters?'' in \emph{Proc. Int. Conf. Intelligent
  Information Systems}, Zakopane, Jun. 2008, pp. 35--44.

\bibitem{Tschiatschek_PrecisionBounds}
S.~Tschiatschek and F.~Pernkopf, ``On {Bayesian} network classifiers with
  reduced precision parameters,'' \emph{{IEEE} Trans. Pattern Anal. Mach.
  Intell.}, vol.~37, no.~4, pp. 774--785, Apr. 2015.

\bibitem{Reznik_Quantization}
Y.~Reznik, ``An algorithm for quantization of discrete probability
  distributions,'' in \emph{Data Compression Conference (DCC)}, Snowbird, UT,
  Mar. 2011, pp. 333--342.

\bibitem{Boecherer_SCC13}
G.~B\"ocherer, ``Optimal non-uniform mapping for probabilistic shaping,'' in
  \emph{Proc. Int. ITG Conf. Systems, Communication and Coding (SCC)}, Munich,
  Jan. 2013, pp. 1--6.

\bibitem{Geiger_OptimalQuantization}
G.~B\"ocherer and B.~C. Geiger, ``Optimal quantization for distribution
  synthesis,'' submitted to \emph{IEEE Trans. Inf. Theory}; preprint available:
  {\tt arXiv:1307.6843 [cs.IT]}.

\bibitem{Cover_Information2}
T.~M. Cover and J.~A. Thomas, \emph{Elements of Information Theory},
  2nd~ed.\hskip 1em plus 0.5em minus 0.4em\relax Hoboken, NJ: Wiley
  Interscience, 2006.

\bibitem{Verdu_ITABound}
S.~Verd\'{u}, ``Total variation distance and the distribution of relative
  information,'' in \emph{Proc. Inf. Theory and Applicat. Workshop (ITA)}, San
  Diego, CA, Feb. 2014, pp. 499--501.

\bibitem{Michalewicz_ModernHeuristics}
Z.~Michalewicz and D.~B. Fogel, \emph{How to Solve It: Modern Heuristics},
  2nd~ed.\hskip 1em plus 0.5em minus 0.4em\relax Berlin: Springer, 2004.

\bibitem{Csiszar_IProjections}
I.~Csisz\'{a}r and M.~Franti\u{s}ek, ``Information projections revisited,''
  \emph{{IEEE} Trans. Inf. Theory}, vol.~49, no.~6, pp. 1474--1490, Jun. 2003.

\bibitem{Rached_KLDR}
Z.~Rached, F.~Alajaji, and L.~L. Campbell, ``The {Kullback-Leibler} divergence
  rate between {Markov} sources,'' \emph{{IEEE} Trans. Inf. Theory}, vol.~50,
  no.~5, pp. 917--921, May 2004.

\bibitem{Heinrich_Multiplier}
L.~Heinrich, F.~Pukelsheim, and U.~Schwingenschl\"ogl, ``On stationary
  multiplier methods for the rounding of probabilities and the limiting law of
  the {Sainte-Lagu\"e} divergence,'' \emph{Statistics \& Decisions}, vol.~23,
  pp. 117--129, 2005.

\bibitem{Boyd_CVX04}
S.~Boyd and L.~Vandenberghe, \emph{Convex Optimization}.\hskip 1em plus 0.5em
  minus 0.4em\relax Cambridge et al.: Cambridge University Press, 2004.

\end{thebibliography}

\end{document}